\documentclass[sigconf]{acmart}
\usepackage{graphicx} 

\usepackage{url}            
\usepackage{booktabs}       
\usepackage{amsfonts}       
\usepackage{nicefrac}       
\usepackage{microtype}      
\usepackage{acmart-taps}         
\usepackage[linesnumbered,ruled,vlined]{algorithm2e}

\usepackage{booktabs}

\allowdisplaybreaks[1]  

\newtheorem{theorem}{Theorem}

\newtheorem{definition}{Definition}
\newtheorem{example}{Example}

\newcommand*{\FQF}{F^{\textsf{QF}}}

\newcommand*{\FCOQF}{F^{\textsf{CO-QF}}}
\newcommand*{\FCOQFp}{F^{\textsf{CO-QF'}}}

\SetKwInput{KwInput}{Input}
\SetKwInput{KwOutput}{Output}
\newcommand{\subs}{\texttt{subsidy}~}

\copyrightyear{2025}
\acmYear{2025}
\setcopyright{acmlicensed}\acmConference[EAAMO '25]{Equity and Access in Algorithms, Mechanisms, and Optimization}{November 5--7, 2025}{Pittsburgh, PA, USA}
\acmBooktitle{Equity and Access in Algorithms, Mechanisms, and Optimization (EAAMO '25), November 5--7, 2025, Pittsburgh, PA, USA}
\acmDOI{10.1145/3757887.3763019}
\acmISBN{979-8-4007-2140-3/2025/11}

\begin{document}

\title{Fair Decisions through Plurality: Results from a Crowdfunding Platform}

\author{Joel Miller}
\affiliation{%
  \institution{University of Illinois Chicago}
  \city{Chicago}
  \country{USA}}
\email{jmill54@uic.edu}

\author{E. Glen Weyl}
\affiliation{%
  \institution{Plural Technology Co-Laboratory, Microsoft Research}
  \city{Seattle}
  \country{USA}}
\email{glenweyl@microsoft.com}

\author{Chris Kanich}
\affiliation{%
  \institution{University of Illinois Chicago}
  \city{Chicago}
  \country{USA}}
\email{ckanich@uic.edu}

\begin{abstract}
    We discuss an algorithmic intervention aimed at increasing equity and economic efficiency at a crowdfunding platform that gives cash subsidies to grantees. Through a blend of technical and qualitative methods, we show that the previous algorithm used by the platform -- Quadratic Funding (QF) -- suffered problems because its design was rooted in a model of individuals as isolated and selfish. We present an alternative algorithm -- \emph{Connection-Oriented Quadratic Funding} (CO-QF) -- rooted in a theory of plurality and prosocial utilities, and show that it qualitatively and quantitatively performs better than QF. CO-QF has achieved an 89\% adoption rate at the platform and has distributed over \$4 Million to date. In simulations we show that it provides better social welfare than QF. While our design for CO-QF was responsive to the needs of a specific community,  we also extrapolate out of this context to show that CO-QF is a potentially helpful tool for general-purpose public decision making.
\end{abstract}


\keywords{Algorithmic Fairness, Quadratic Funding, Public Goods, Plurality}

\maketitle

\section{Introduction}

Public goods funding is traditionally seen as a difficult problem in economics because wholly self-interested agents cannot fund public goods within usual market frameworks -- an idea reflected in the classic ``Tragedy of the Commons'' thought experiment~\cite{samuelson1954pure, hardin1968tragedy}. A mechanism called \textit{Quadratic Funding} (QF) addresses this problem, yielding outcomes that maximize utilitarian social welfare under typical economic assumptions. 
Quadratic Funding was implemented in practice by a crowdfunding and grants platform called Gitcoin,\footnote{\texttt{gitcoin.co}}\footnote{Although its name is a portmanteau of ``Git'' and ``Bitcoin'', Gitcoin is not formally affiliated with either.} which has used QF to award money to open-source software projects and other public goods since 2018. 

While exhibiting several benefits in theory, QF suffered problems with both economic efficiency and fairness in practice~\cite{gitcoin2022how}. In this paper, we describe the design, implementation, and operation of  \textit{Connection-Oriented Quadratic Funding} (CO-QF), an algorithm we created to address the problems QF faced. Since its introduction one and a half years ago, CO-QF has been used in 89\% of Gitcoin's grant-awarding ventures, distributing just over \$4 million USD.

We designed this algorithm in collaboration with the community, ensuring that it met the practical needs of stakeholders. In parallel, we developed a technical framework that could explain why QF faced the problems it did. We propose a model of \textit{prosocial utilities} as the best way to make sense of QF's issues from a mathematical perspective. CO-QF's design is rooted in this theory, and also draws inspiration from the normative concept of plurality. 

To understand the community's response to CO-QF we conduct interviews with stakeholders, showing that CO-QF was exceedingly favored as a distributive tool and seen as more fair. On the quantitative side, we use simulations to show that CO-QF provides better welfare guarantees than QF in the presence of prosocial utilities.



While our qualitative data are coupled to the unique circumstances of the Gitcoin community, we also provide a set of algorithmic design insights that are germane to CO-QF's general capabilities across wider contexts. Ultimately, the data suggest that CO-QF's alignment with prosocial outcomes, quantitative improvements in welfare, and qualitative community approval combine to constitute a promising direction for all manner of fair public decision making. 

\subsection{A Sociotechnical Approach}
Our paper draws on a line of computer science research that sees algorithms as ``heterogeneous and diffuse sociotechnical systems, rather than rigidly constrained and procedural formulas''~\cite{seaver2017algorithms}. Instead of simply focusing on the mathematical implications of a particular modeling assumption (prosocial utilities) on a particular mechanism (Quadratic Funding), we're also interested in understanding the Gitcoin community's varied perceptions of and ways of interacting with QF and CO-QF, 
enacting an intervention that accounts for ``the social and political contexts of specific, situated technical systems at their points of use''~\cite{katell2020toward}. 

This type of qualitative analysis is important firstly because ``over-relying on data-driven decision-making to inform policymaking commonly leads to proposed solutions that only superficially address societal problems,'' as Russo~\emph{et al} report~\cite{russo2024bridging}. In other words, a firm grasp of how the community practically concieves of QF and CO-QF is essential for creating a practically useful algorithm.  

Secondly, studying Gitcoin's social reality helps us understand CO-QF's broader applicability as a public decision making mechanism \textit{outside} of the community. As we will see, Gitcoin community members evaluated QF and CO-QF on several different axes. When CO-QF was straightforwardly evaluated as a public goods funding mechanism, it was highly favored. However, Gitcoin's unique context led community members to evaluate the algorithm in other ways as well. Disentangling which qualitative data are more specific to CO-QF's instantiation at Gitcoin, and which are less specific, is essential for understanding CO-QF's wider applicability. This line of analysis corresponds to Seaver's conception of algorithms as ``multiples''\cite{seaver2017algorithms}: we understand QF/ CO-QF not just as one process, but as loci for several inter-related patterns of behavior.

\subsection{Related Work}
For a general overview of public goods funding systems (excluding QF), see~\cite{batina2005public}. Quadratic Funding was first introduced in~\cite{buterin2019flexible}. Along theoretical lines, Freitas and Maldonado analyze QF under imperfect information~\cite{freitas2025quadratic}, while Li \emph{et al} study the effects of expert advice on funding outcomes~\cite{li2024decentralized}. Pasquini presents preliminary work analyzing the effects of bounded budgets on QF~\cite{pasquini2020quadratic,pasquini2022optimal} from both a mathematical and empirical perspective via statistics from Gitcoin. Quarfoot \emph{et al} and Cheng \emph{et al} show that using Quadratic Voting-like mechanisms on surveys yields more accurate understandings of user preferences~\cite{quarfoot2017quadratic,cheng2021can,cheng2025organize}. Dimitri provides an economic analysis of blockchain-based QV systems~\cite{dimitri2022quadratic}, while Wright provides a legal analysis~\cite{wright2019quadratic}. Kelter \emph{et al} study large scale agent-based simulations of quadratic votes~\cite{kelter2021agent}.

Nabben performed a qualitative digital ethnography of the Gitcoin community~\cite{nabben2023governance}, focusing on its organization structure and governance. 
We also provided an initial report on CO-QF's implementation and experiences using it at GitCoin~\cite{miller2024case}. For a general literature review of research on blockchain-based organizations (e.g. Gitcoin), see~\cite{bonnet2024decentralized}.

Other work has examined the economics of crowdfunding platforms that use direct donation mechanisms instead of QF. For example, Hong and Ryu explore the potential of these platforms to fund public goods with the help of the public sector~\cite{hong2019crowdfunding}, and Corazzini \emph{et al} study donor coordination~\cite{corazzini2015donor}.  

\subsection{Outline}
In Section~\ref{sec:model} we lay out a mathematical model of public goods funding. In Section~\ref{sec:interview methodology} we present our interview methodology. In Section~\ref{sec:QFs beginnings} we contextualize QF, explaining an important mathematical result and describing how the algorithm was adopted for use at Gitcoin. In Section~\ref{sec:QF problems} we explore issues with QF: first via a a quantitative analysis, and secondly via qualitative interview data reflecting the same issues suggested by our quantitative analysis. In Section~\ref{sec:COQF design} we explain the design principles behind CO-QF and explore how the algorithm reflects a pluralistic viewpoint. In Section~\ref{sec:COQF response} we describe the community's response to CO-QF through more qualitative data, showing that CO-QF was favored as a distributive mechanism. In Section~\ref{sec:simulations} we return to a quantitative standpoint and show simulations in which CO-QF is more economically efficient. We conclude with a discussion in Section~\ref{sec:discussion}.

\section{Mathematical Model}\label{sec:model}

We consider the issue of public goods funding. For notational simplicity, we will focus on the issue of deciding funding for a single public good. 

Let $N = \{1,\dots,n\}$ be a set of agents. We imagine a scenario where agents  contribute currency to a central mechanism, which decides how much funding to give to the public good. Let $c_i$ denote agent $i$'s monetary contribution to the mechanism. Our paper is broadly concerned with changes to an algorithm called Quadratic Funding~\cite{buterin2019flexible}.
\begin{definition}[Quadratic Funding (QF)]
    Given contributions $(c_i)_{i \in N}$, Quadratic Funding awards a public good 
    \[
    \FQF = \biggl(\sum_{i \in N} \sqrt{c_i}\biggr)^2
    \]
    in funding.
\end{definition}

With two or more contributors the funding awarded by QF is greater than the sum of individual contributions, so QF relies on the assumption that the mechanism designer has a pool of subsidy cash to award to public goods.\footnote{In other models, agents are taxed to make up for the difference, but we do not consider that scenario here.}

We assume each agent $i$ has a ``personal'' utility function $u_i(F)$ describing the currency-equivalent utility they receive when the public good receives funding level $F$. We assume each $u_i$ is concave, smooth, and increasing.

We also assume that agents exhibit prosocial behavior, i.e., they take the prospective well-being of others into account when making decisions. We employ a generalization of the longstanding ``sympathy coefficient'' model first posited in 1881 by Edgeworth~\cite{edgeworth1881mathematical} and later used and adapted by many other economists, e.g.~\cite{becker1974theory,bruce1990rotten,charness2002understanding,tilman2019localized,flanigan2023distortion}. 
\begin{definition}[Prosocial Utility Functions]
    Let $\alpha_{ij} > 0$ describe the sympathy coefficient of agent $i$ towards agent $j$. Then agent $i$'s prosocial utility function is 
    \[
    \hat{u}_i(F) = \sum_{j\in n}\alpha_{ij} \cdot u_{j}(F)
    \]
\end{definition}
Setting $\alpha_{ii} = 1$ and $\alpha_{ij}=0$ for all $i \neq j$ recovers selfish utility functions, and setting $\alpha_{ii} = 1$ and setting $\alpha_{ij}$ to the same constant for all $i \neq j$ recovers Edgeworth's original sympathy coefficient model.

Lastly, we discuss social welfare. In this paper we specifically study utilitarian social welfare, which is the benefit afforded to all members of society from the public good, minus the funding cost. 
\begin{definition}[(Utilitarian) Social Welfare]
    Suppose a public good is funded at level $F$. Then utilitarian social welfare is 
    \[
    USW(F) = \sum_{i \in N}u_i(F) - F
    \]
\end{definition}
Notice that we use the individual utility functions $u_i$, rather than the prosocial utility functions $\hat{u_i}$ in our definition of social welfare. This normative choice is common in the literature~\cite{hammond2018altruism,harsanyi1977morality}. It avoids double-counting utilities and reflects the idea that individuals may use $\hat{u}_i$ when making decisions (say, about how much to donate), but ultimately the best way to measure benefit to society is still to consider how much each individual benefits on their own.

\section{Interview Methodology}\label{sec:interview methodology}

In this section, we explain the methodology driving our qualitative interviews. The following study design passed a formal Institutional Review Board (IRB) process at an accredited U.S. institution.

\paragraph{Participants.}
Our participant pool consists of eight ``round managers'' (P1-P8) who were not employed at Gitcoin and volunteered their time to help manage funding rounds. We chose to interview this subset of the community to reduce sbias, since 1) Gitcoin employees may have been more inclined to speak in a way that highlighted the strengths of the organization and 2) grantees or donors may have been more inclined to favor the algorithm that gave more money to their preferred project(s). Although we designed CO-QF with feedback from the Gitcoin community, we elected to only interview individuals that we had not spoken with before. Participants were mostly male (7/8) and hailed from a range of locations (North America, South America, The Caribbean, Africa, Europe, and Asia).

\paragraph{Interview Structure.}
We conducted semi-structured video interviews with P1-P8. In our interviews, we first asked participants to give background on their relationship with Gitcoin and the surrounding community, and then asked them to reflect on what they liked and disliked about QF and CO-QF. We opted to ask simple, open-ended questions about QF and CO-QF so as to not bias interviewees towards thinking any specific issues or benefits were more important. 
Interviews ranged between 15 and 60 minutes in length.

\paragraph{Coding and Analysis.}
We used inductive thematic analysis \cite{braun2006using} to code the audio from the interviews. All coding was done by the first author, with review and approval from the other author(s). After the initial coding step, we used an axial coding process to understand broader themes in the data.

\paragraph{Limitations.} The pool of possible interviewees (round managers not employed by Gitcoin) is relatively small, limiting the amount of data we could collect. Our tactic of asking open-ended questions also has a limitation: if a participant \textit{did not} bring up a specific appraisal of a mechanism, then we cannot be sure if that appraisal was fully irrelevant to them or if they just didn't think to bring it up during the interview. 

\section{QF's theoretical development and implementation at Gitcoin}\label{sec:QFs beginnings}

Quadratic Funding was introduced in 2019 by Heitzig, Buterin and Weyl~\cite{buterin2019flexible}. It shares theoretical underpinnings with Weyl and Lalley's Quadratic Voting mechanism~\cite{lalley2018quadratic}.\footnote{Going back further, QV/ QF also build on work by 
Grovers and Ledyard~\cite{groves1977optimal} and Hylland and Zeckhauser~\cite{hylland1979mechanism}.}
While classical results from economics cast doubt on the ability of markets (i.e. direct donation schemes) and 1-person-1-vote schemes to optimally fund public goods, QF achieves optimally under a set of assumptions common in the economics literature~\cite{buterin2019flexible}. QF optimally funds a public good in the following sense: when selfish agents are allowed to repeatedly adjust their donations in response to the donations of others, the resultant game has a single Nash equilibrium where $\FQF$ maximizes utilitarian social welfare~\cite{buterin2019flexible}. 

Soon after QF's introduction, a crowdfunding and grants platform named Gitcoin started using it to award money to open-source software projects and other public goods. Gitcoin receives large donations from software foundations, NGOs, and corporations, with the mandate that the money be awarded to public goods, thus supplying the platform with the pool of subsidy cash required to run QF. 
Gitcoin's adaptation of QF features \textit{funding rounds} which admit a fixed set of (manually vetted) projects, have a fixed pool of subsidy cash to be awarded (usually \$50k-\$200k), and run for a set time period. Within that time period, users are free to donate any amount to any projects they see fit. After the end of the round, QF is used to obtain a ``raw'' subsidy amount for each project. This subsidy amount is calculated on a per-project basis as $F^{\textsf{QF}} - \sum_i c_i$, or equivalently via algorithm ~\ref{alg:QF}. In other words, a project's raw subsidy amount is the amount of funding QF ``adds on'' on top of the direct donations. Raw subsidy amounts are then normalized so that their sum matches the size of that round's subsidy pool.\footnote{Most rounds also implement a ``matching cap'' which prohibits any project from taking more than some percent of the subsidy pool -- if so, funding in excess of this amount is re-distributed to other projects below the matching cap.}
In total, each project is awarded its direct donations plus its portion of the subsidy pool.

\begin{algorithm}
    \KwInput{$c_{i}$ for $i \in N$ (contributions to the project)}
    \KwOutput{\subs (the un-normalized amount of subsidy funding awarded to the project)}
    $\subs \gets 0$; \\
    \For{$i \in N$}{
        \For{$j \in N$}{
            \If{$i \neq j$}{
                $\subs \gets \subs + \sqrt{c_{i} \cdot c_{j}}$;
            }
        }
    }
    \Return $\subs$
    \caption{Calculation of QF subsidy amounts for a single project}\label{alg:QF}
\end{algorithm}

Gitcoin runs its own funding rounds, and community members and external organizations also run funding rounds using Gitcoin's infrastructure. Each round has one or more ``round managers'' who oversee that round's operation, and who may or may not be Gitcoin employees (our interviewees are non-employee round managers). Many rounds are jointly managed by a combination of employees and non-employees. 

\section{Theoretical and Practical Issues with QF}\label{sec:QF problems}

While QF had benefits over other schemes and afforded the platform public interest, the algorithm also suffered from problems with equity and accuracy. We will explain this issue from both the theoretical and practical perspective.

Theoretically, QF derives its optimality (in part) from the assumption that all individuals are selfish and that, outside of the public good under consideration, all their consumption is of private goods. However, if individuals are social because of altruism, coordination or because beyond the present application they participate in networks of social consumption, QF loses its optimality. The following theorem makes this explicit.

\begin{theorem}
    Suppose $\alpha_{ii} = 1$ for all $i$, and $\alpha_{ij} > 0$ for at least one pair $p,q$ with $p \neq q$ and positive $u_q$. Then $\FQF$ does not maximize social welfare.
\end{theorem}
\begin{proof} Let $u(F) = \sum_i u_i(F)$.
Agent $i$'s contribution will be chosen to maximize
\[
\hat{u_i}\biggl(\biggl(\sum_{j \in N} \sqrt{c_j}\biggr)^2\biggr) - c_i
\]
which in equilibrium will have to satisfy
\begin{align}\label{math:qf pf}
    \frac{\hat{u_i}'(\FQF) \bigl(\sum_{j \in N} \sqrt{c_j}\bigr)}{2 \sqrt{c_i}} = 1 \Leftrightarrow \hat{u_i}'(\FQF) = \frac{\sqrt{c_i}}{\sum_{j \in N}\sqrt{c_j}}
\end{align}
by differentiation. Then we have
\begin{align*}
1 = \sum_{i \in N}\hat{u_i}'(\FQF) \geq \sum_{i \in N}u'_i(\FQF) + \alpha_{pq}u_q(\FQF) 
\Leftrightarrow u'(\FQF) < 1
\end{align*}
where the first equality is attained by summing Eq.~\ref{math:qf pf} across all agents.
But since all $u_i$s are smooth, concave, and increasing, social welfare is only maximized when $u'(\FQF) = 1$.
\end{proof}

Notice that in the above proof, QF specifically loses its optimality by \textit{overshooting} the optimal funding amount. This leads to a fairness issue in the context of a bounded subsidy pool, since large groups of agents who value the same good and have pro-social utilities for each other are able to ``convince'' the mechanism that their project matters more, draining subsidy from other projects. The following stylized example makes this issue explicit:

\begin{example}
    Let $g(F)$ be a smooth, convex, and increasing function. Suppose there are two separate public goods which can be funded. Let $F_1$ and $F_2$ denote the funding levels for these goods. Suppose a set of $n$ agents have the following utility functions over funding to the two goods:
    \begin{align*}
        u_i(F_1,F_2) = 
        \begin{cases}
            g(F_1) &\hbox{if}\;\; i \in\{1,2\}\\
            g(F_2) &\hbox{if}\;\; i =3\\
            0      &\hbox{if}\;\; 4 \leq i \leq n
        \end{cases}
    \end{align*}
    Furthermore set $\alpha_{i3}=1$ for all $i \geq 4$ and set all other $\alpha$ values equal to 0. 
    Total social welfare is 
    \begin{align*}
        2g(F_1) +g(F_2)
    \end{align*}
    So in the context of a bounded subsidy pool which is not large enough to fund both goods optimally, any distribution that gives more to $F_2$ than $F_1$ is inefficient. However, the existence of agents $4,\dots,n$, who all pro-socially donate on behalf of agent 3, will skew funding results towards $F_2$. In the limit as $n\rightarrow\infty$, the entire subsidy pool will be awarded to $F_2$.
\end{example}

    
    

The theory suggests, then, that in the context of a bounded subsidy pool, a large group of coordinated agents can shift funding outcomes in their favor and \textit{draw funding away from other projects}. This presents an equity issue in the sense that smaller groups are deprived of a voice in the funding process, even if those groups would stand to benefit from funding to their preferred projects. As our interviewees explain, this is also what happened in practice, as the incentives created for grantees turned QF into what many of them called ``a popularity contest''.

\subsection{``A Popularity Contest'': Qualitative Opinions of QF at Gitcoin}

When asked to open-endedly reflect on what they liked and disliked about QF, 6 out of 8 interviewees (P1-P3, P5-P7) indicated that Gitcoin's implementation of QF had the unfavorable tendency to disproportionately rewarded the most popular projects. Indeed, Gitcoin's donor base participates in a range of shared goods in communities beyond the reach of the platform, allowing for the type of coordination discussed above. According to P2:
\begin{quote}
\emph{
    I think the biggest complaint we hear with vanilla QF is it's more of a ``popularity contest''. So [if a project receives a lot of matching] is that a signal of popularity? Probably to some extent. Is that a signal that projects went out and did amazing things in the world and people think that they're important and should do more of it? Probably some of that too... at the end of the day, [it's] a little difficult to know what that signal means. 
}
\end{quote}
P3 corroborated:

\begin{quote}\emph{
    There is of course the classic thing with quadratic funding which is that if you're a project that has lots of friends, lots of people in Web3,\footnote{``Web3'' is a colloquial term for the larger blockchain-centric community Gitcoin operates in.} you can slightly farm the mechanism in some ways by having an active supporter base that's constantly continually donating to the rounds.
}\end{quote}
%
Three interviewees (P2, P6, P7) specifically indicated that QF's tendency to award the most popular projects constituted a fairness issue.

It is worth exploring why  interviewees took umbrage with QF's tendency to award the most popular projects. After all, in most democratic procedures, it is normal for the most popular option(s) to be favored by the mechanism (e.g., in an election, the most popular candidate wins). What is the difference here?

For P2 (quoted above), the issue was that a project's popularity is not always correlated with its quality. P5 felt similarly, further arguing that limited resources on the part of donors led this dynamic:

\begin{quote}\emph{
    People have limited resources, how do they decide that you are one of the projects in that basket... [donors] don't have the time to go and read [all the project descriptions] and then decide, so it ends up being a popularity contest in some sense -- how many people are able to recall your project's name?
}\end{quote}

P3 echoed this sentiment, explaining that 

\begin{quote}\emph{
    [Donors] go through the list and pick out the projects they either know or are connected to... rather than doing the rational economist view of how you want people to donate, which is to like really evaluate which is the best project that deserves the funding... Some [projects] are maybe a little bit less serious in terms of the talent that's actually behind them... they're more talented at marketing than they may be as seriously hardcore talented founders and builders with really well thought out projects.
}\end{quote}
These quotes help to illuminate the distortive dynamics created by advertising in the Gitcoin ecosystem. Since Gitcoin awards money to grantees with no strings attached, they may have an incentive to attain more funding for their project regardless of its value to others. This incentive, combined with the social nature of the Gitcoin ecosystem, leads to severe inequity: for  a grantee looking to attain more funding, marketing (i.e. attempting to alter the prosocial utilities of others) may be easier than building better services for the community. The result is that projects with more advertising resources (i.e., the ``popular'' ones) can dominate over projects that provide more utility to individuals.


On the other hand, not all of our interviewees took a completely negative view of these dynamics. While P5 did acknowledge the ``popularity contest'' dynamic, they also argued that

\begin{quote}\emph{
    The counterpoint is it's only a popularity contest for some time. I can go and get a lot of noise and votes one or two times... but if I'm not delivering and people haven't seen my project grow... my experience tells me, I've seen projects fall down. You could be very loud, but there is a correcting force that comes in. 
}\end{quote}

P8 also voiced concerns about the way large numbers of donors effect QF results, but framed the issue in terms of Sybil attacks, fraud, and other technical weaknesses. Gitcoin does employ other technical tools to limit Sybil attacks and other types of fraud, but QF is especially vulnerable to these types of fraud for the same reason it is vulnerable to the effects of prosocial utilities at large.\footnote{E.g., one could see a Sybil agent as an agent with strong pro-social utility toward its creator.}
P4 was our sole interviewee with nothing negative to say about QF:
\begin{quote}\emph{
    I understand quadratic funding to be a useful way of harnessing the wisdom of the crowd to determine fund allocation... I think it's it's also great just for participatory budgeting and fostering a community of people feeling like their voice is heard and that can sort of help cohere a community and an ecosystem.
}\end{quote}
However, most of our interviewees indicated that QF had a problematic tendency to excessively award large, well-marketed projects.

\section{CO-QF: A Plural Alternative}\label{sec:COQF design}

\textit{Connection Oriented Quadratic Funding} (CO-QF), our modification to QF, is based on the principle of plurality or "cooperation across difference" ~\cite{weyl2024plurality}. 
In other words, a plural approach to public goods funding prioritizes projects that are appreciated by diverse sets of participants, rather than homogeneous groups (i.e. groups whose members ostensibly have strong prosocial utilities towards each other). We hope the intuition behind this design direction is clear: if prosocial utilities create problems for QF, then prioritizing the inputs of agents who don’t have strong prosocial utilities for each other might ameliorate those problems.\footnote{Our choice to embrace the normative concept of plurality was inspired by Green and Viljoen's argument for ``Algorithmic Realism'': among other things, they argue that since algorithms can never be truly neutral, one might as well think carefully about the normative assumptions behind one's work and choose a justifiable and appropriate direction.~\cite{green2020algorithmic}}

However, in the context of practical use at Gitcoin, the social information we have access to is limited. While one could write down an algorithm that, say, directly referenced $\alpha_{ij}$ values in its calculations, how values could be collected is unclear. $\alpha_{ij}$ values cannot be known \textit{a priori} are in any case merely an modeling abstraction meant to explain a wide range of phenomena, not a directly measurable quantity. 
Furthermore, attempting to estimate $\alpha_{ij}$ values via additional 
user input risked compromising the fluidity of the user experience. 
%
We instead opted for a simpler approach to understanding prosocial behavior.

Specifically, CO-QF hinges on a conception of prosocial utilities inspired by the sociologist Georg Simmel. Simmel saw identity as being defined by one's memberships in social groups: for him, ``the groups to which an individual belongs form, so to speak, a system of coordinates, in such a way that each newly added group determines the individual more precisely and more unambiguously''~\cite{Simmel}. CO-QF uses mappings between individuals and the (relevant) groups they are members of as input, instead of $\alpha_{ij}$ values. We assume that one's group membership has some correlation with their $a_{ij}$ values (although as we show in Section~\ref{sec:simulations}, CO-QF works even when this correlation is not strong).

Formally, assume we have a set $G \subseteq 2^N$ of groupings of agents. Each element $g \in G$ is a subset of $N$, and for each $i \in g$, we assume there exists a weight $w_{i,g}$ describing the strength of $i$'s membership in $g$, normalized so that $\sum_{g \in G} w_{i,g} = 1$. CO-QF implements a plural perspective by awarding more funding to projects liked by \textit{different pairs social groups}, instead of projects liked by different pairs of individuals (as is the case in normal QF, and made explicit in Algorithm~\ref{alg:QF}). 
See algorithm~\ref{alg:COQF} for a technical description of how CO-QF calculates subsidy amounts.
 
\begin{algorithm}
    \KwInput{$c_{i}$ for $i \in N$ (contributions to the project), $G$ (social groups), and $w$ (weights in social groups)}
    \KwOutput{\subs  (the un-normalized amount of subsidy funding awarded to the project)}
    $\subs \gets 0$; \\
    \For{$g \in G$}{
        \For{$h \in G$}{
            \texttt{g\_sum} $\gets 0$;\\
            \texttt{h\_sum} $\gets 0$;\\
            \For{$i \in g \setminus h$}{
                $\texttt{g\_sum} \gets \texttt{g\_sum} + c_{i} \cdot w_{i,g}$;
            }
            \For{$j \in h \setminus g$}{
                $\texttt{h\_sum} \gets \texttt{h\_sum} + c_{j} \cdot w_{j,h}$;
            }
            $\subs \gets \subs + \sqrt{\hbox{\texttt{g\_sum}} \cdot \hbox{\texttt{h\_sum}}}$
        }
    }
    \Return $\subs$
    \caption{Calculation of Connection-Oriented QF (CO-QF) subsidy amounts for a single project}\label{alg:COQF}
\end{algorithm}

Informally, CO-QF calculates subsidies as follows. For each pair of social groups $g$ and $h$, we check if members of \textit{just one group, but not the other} have contributed to the project. A project's subsidy only increases if agents on both sides of this symmetric difference support the project -- in other words, if there is \textit{agreement across difference} on the merit of the project (where ``difference'' is relative to $g$ and $h$). In this way, projects with support from a plurality of social groups are prioritized. When summing contributions from members of a group, we multiply each contribution by that member's weight in the group to prevent members of many groups from having an outsized impact on funding amounts. Note that setting $G = \{\{i\}\;|\;i \in N\}$ recovers QF.

Algorithm~\ref{alg:COQF} actually reflects the second iteration of CO-QF. The first iteration had stronger properties from a mathematical standpoint but was unsuccessful in practice. See Appendix~\ref{apdx:first design} for details.

The inputs $G$ and $w$ are fairly general abstractions for information about social structure. For example, in a firm using CO-QF to make a decision, $G$ could divide the firm into its different departments. In a municipality using CO-QF, $G$ could indicate neighborhoods and/or tax brackets. In Gitcoin's case, finding an appropriate choice of $G$ was difficult due to the lack of data on donors (see Appendix~\ref{apdx:otherG} for details). Ultimately we decided to use the set of projects as the set of social groups. If donor $i$ donated to project $p$, they would be placed into project $p$'s group, with a weight corresponding to to the amount they donated to $p$ (relative to their entire amount donated). Despite the simplicity of this approach, it was successful in practice.

\section{Community Response to CO-QF}\label{sec:COQF response}

Before diving into the qualitative data, we will first give some high-level statistics concerning CO-QF's adoption at Gitcoin. As we introduced CO-QF, we retained the option for round operators to use QF, and added an option to use a hybrid mechanism blending CO-QF and QF results.\footnote{The hybrid approach simply took a weighted average of the subsidy amounts from each mechanism on a per-project basis. Round operators could weight CO-QF and QF equally, or they could choose a 25\%-75\% split in either direction.} 
Since CO-QF's introduction in late 2023, Gitcoin has facilitated 38 funding rounds, through which a total of around \$4.5 Million has been distributed. Of these 38 rounds, 32 have used CO-QF or a hybrid mechanism, giving away around \$4 Million in funding combined. See Figure~\ref{fig:adoption} for a more detailed breakdown of CO-QF's adoption metrics.
\begin{figure}
    \centering
    \includegraphics[width=\linewidth]{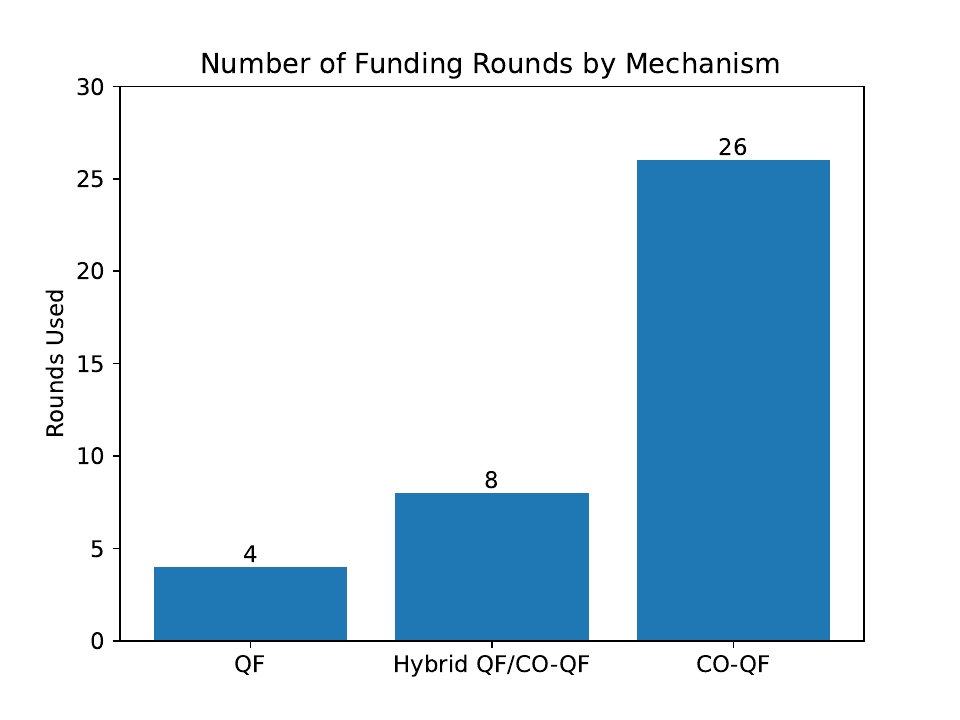}
    \includegraphics[width=\linewidth]{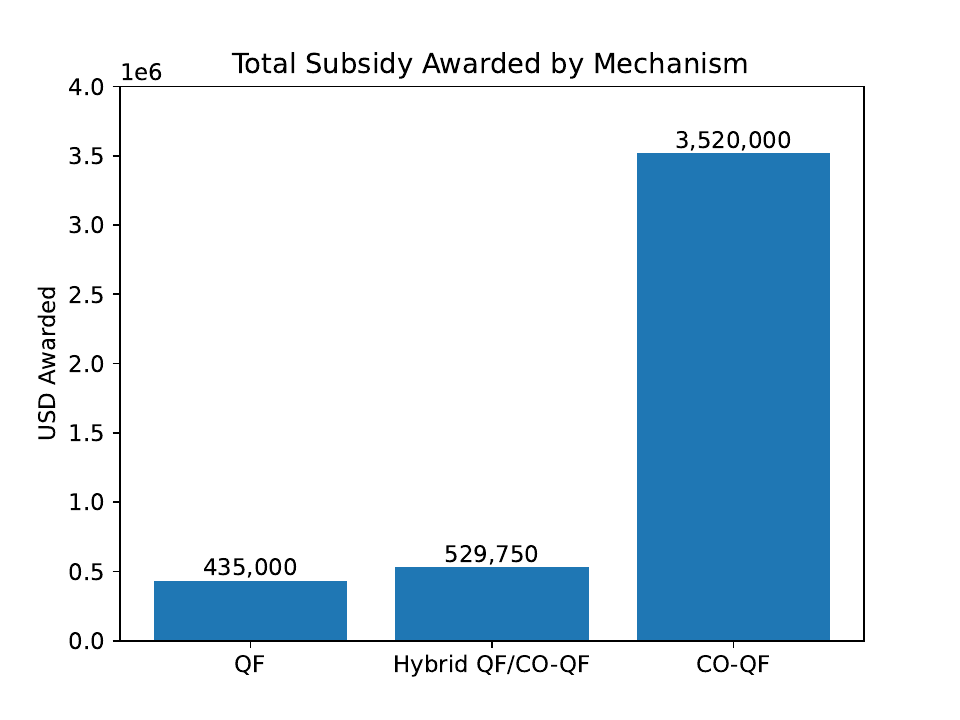}
    \caption{Number of rounds using each mechanism (top) and total subsidy awarded by each mechanism (bottom) since CO-QF's introduction. Subsidy amounts are in US Dollars.
    }
    \label{fig:adoption}
\end{figure}
These statistics indicate that CO-QF was successful in solving the needs of round managers, who chose to use CO-QF or a hybrid mechanism significantly more than they chose to use QF. 

Our interviews paint a richer picture of CO-QF's reception at Gitcoin. Interveiwees tended to evaluate CO-QF\footnote{In the Gitcoin community, the algorithm we call CO-QF is known by a few different names (COCM, Cluster Matching, Cluster Mapping, and others). In our interview transcriptions we have replaced instances of these names with ``CO-QF'' for clarity.} along two different axes: as a tool for distributing funding, and as a locus for sparking community engagement.

\subsection{CO-QF as a distributive tool}

All of our interviewees felt that CO-QF was an improvement over QF as far as its distributive capacities were concerned. 5 out of 8 interviewees (P1, P2, P5-P7) specifically indicated that CO-QF provided a clearer signal of which projects would benefit their community as a whole. P5 said, 

\begin{quote}\emph{
I think the main objectives were that people should not be able to fake support and there should be some actual wisdom of the crowds... to a great extent, CO-QF does succeed [in those respects]...
What we're trying to optimize for is that the right projects get the money. Not the ones who look right, but the ones that are right, and that difference is a pretty big one. The more we're able to lean towards the ones that are actually delivering, the more efficient the capital allocation that's happening. And CO-QF cuts through the noise. 
}\end{quote}
P1 gave the following evaluation:
\begin{quote}\emph{
    Well I think it was helpful in reducing fraud, if anything. Or, not necessarily fraud, but in terms of getting to the optimal -- like, ``fair'' is probably not the best word either, but getting to the most accurate depiction of what the total voting pool wanted most. So in our sense it was like, we want to get the best picture of what the community would get the best benefit from. So that's where CO-QF helped with individual projects that maybe were grouping in votes from single sources or just outside sources that didn't really encompass the true sentiment of the broader group.
}\end{quote}

P1's answer shows how different negative qualities -- vulnerability to fraud, unfairness, and failure to reflect broader group sentiment -- all can cohere in practice. On the other hand, P4, P6, and P7 specifically named fairness as a positive quality of CO-QF's distributions. P4 praised CO-QF's outcomes as ``\emph{fair}'' and ``\emph{equitable}''. As P7 explained,

\begin{quote}\emph{
    The main differnce between [CO-QF] and QF, is it's more fair, because it recognizes the people who are building, and those people need to have more voice in the space. 
}\end{quote}
P3 did not elaborate on why they liked CO-QF, simpy stating ``\emph{We're very pleased with CO-QF. It worked very well and I liked it}''. On the other hand, P7 pointed out a flaw in CO-QF's distributive abilities, which stems from our choice to use the set of projects as $G$:

\begin{quote}\emph{
    My understanding is that if a bunch of voters vote for the same set of projects, that is not given enough weight [under CO-QF]... if people have been seeing how projects have been delivering round after round, and if multiple people come to the same conclusion independently... if that happens -- and that has happened -- people who come to that conclusion start getting penalized. 
}\end{quote}

P8, who criticized QF for its vulnerability to fraud and Sybil attacks, felt that CO-QF ameliorated these problems. While we did not specifically design CO-QF with the primary goal of fraud/ Sybil attack prevention,\footnote{Gitcoin employs other tools to ameliorate these issues.} it does reduce the effects of some basic fraud and Sybil attacks patterns. P8 indicated that CO-QF was helped align grantee behavior with the goals of the round managers:
\begin{quote}
    \emph{We started telling our grantees, be aware of CO-QF -- rather than do the wrong things, just do the right thing and you will get [success], so don't try to game the system.
    }
\end{quote}

\subsection{CO-QF as a locus for community activity}

While interviewees praised CO-QF's distributive ability, they were mixed on CO-QF's function as an engine for community building. Indeed, while QF may have overshot optimal subsidy amounts, the promise of large subsidies was a driver of both donor and grantee activity. Donors felt excited to because ``\textit{a dollar to a project might mean they end up generating \$100 [in subsidy]}'' (P2) under the QF formula. Grantees felt excited to advertise in pursuit of more funding. Buzz around QF cemented Gitcoin as a hub for the broader community. Reflecting on Gitcoin's purpose, P3 remarked

\begin{quote}\emph{
    One of the core points is not optimal capital allocation, it's community building, the marketing, the collaborations that emerge... There are many collaborations that would not have happened without the rounds. 
}\end{quote}
P8 corroborated, saying ``\emph{the thing I like about QF is the community rally that it brings}''.
In this regard, CO-QF's introduction constituted a drawback. Since advertising was now a less effective strategy, and since an individual donation was no longer guaranteed to provoke subsidy increases for a project, half of our interviewees (P1, P2, P5, P7) worried that there was now less incentive to participate in the platform. P5 explained that ``\emph{There needs to be some amount of active onboarding also... [CO-QF is] deterring the overall growth}''. P5 went on to elaborate that since grantees ``\emph{optimize for number of votes and then end up feeling dejected}'', CO-QF threatened to slow down activity in their ecosystem. P1, P2, and P7 offered similar opinions. 

P6 and P8 offered a different perspective on CO-QF's effect on community growth. P6 noted that CO-QF gives more subsidies in response to donors who give to a diverse set of projects,\footnote{This is not true in general, but it is true when $G$ is set to be the set of projects.}
and therefore reasoned that ``\emph{if I need to donate to five projects for my matching to be counted then more donations are being driven by CO-QF, and that's great for everyone}''.
According to P8, ``\emph{CO-QF helps [grantees] look beyond their own projects}'', thereby increasing community engagement.

The partially negative evaluation of CO-QF's effect on community growth reveals the multiple  simultaneous ways community members conceive of the algorithm, and the different desiderata that come with each viewpoint. On one hand, as a distributive tool it is desirable for the algorithm to be equitable and/or allocate capital efficiently, which (for Gitcoin) necessitated limiting the effect of advertising on funding outcomes. On the other hand, the algorithm is also a means for generating engagement. From this perspective it is desirable for the algorithm to create incentives for grantees to widely advertise their projects. In other words, for the algorithm to be a good community engagement tool it may be desirable to \textit{increase} the effect of advertising on funding outcomes. The introduction of CO-QF, which improved on the former of these desiderata at the expense of the latter, shows how different conceptions of an algorithm can create conflicting goals for a community. 

The unique needs of communities need to be considered when designing algorithmic interventions \cite{russo2024bridging,green2020algorithmic,katell2020toward}, so negative evaluations of CO-QF's effect on community engagement led us to allow round managers to still use QF if they wanted, or to blend the results of the two algorithms together (the ``hybrid'' approach mentioned above).  

On the other hand, we are also interested in discovering general principles about mechanism design for equitable public decision making. So, while it is problematic for the Gitcoin community that grantees have less of an incentive to advertise under CO-QF, in the wider context of public goods funding mechanisms it is not clear whether this would constitute an issue. On the other hand, the positive evaluations of CO-QF's distributive and welfare-increasing abilities seem likely to carry over to other contexts. This is because (to the best of our knowledge) only some communities want a mechanism that drives community engagement, but many communities do want a mechanism that has positive effects on welfare or equity. Regardless, CO-QF helps to reveal a new dimension for algorithm design in which there is a tradeoff between incentives for community activity and anti-collusive properties.

\section{Simulations}\label{sec:simulations}

In Section~\ref{sec:QF problems} we showed that QF does not maximize utilitarian social welfare in the presence of prosocial utilities. In sections~\ref{sec:COQF design} and~\ref{sec:COQF response}, we explored an algorithmic intervention grounded in the needs and constraints a specific community. While our qualitative results indicate that CO-QF was successful in practice, the question still remains: does CO-QF provide \textit{quantitatively} better welfare in the presence of prosocial utilities? We answer that question affirmatively in this section. In the following simulations we explore equilibrium funding outcomes under QF and CO-QF, measuring how well each algorithm approximates the maximum social welfare. We also benchmark QF and CO-QF against a more standard direct donation mechanism where $F$ is simply the sum of the $c_i$s. 


We simulated donations from $n=25$  agents split into $5$ groups of $5$ agents each. Unlike in the case of Gitcoin, where multiple public goods compete for portions of a limited funding pool, we consider the case of just one public good and a potentially unlimited funding pool in order to match the models built to describe QF. 

Each agent has a utility function of the form $u_i = \beta_i\cdot\ln(F+1)$ for some $\beta_i \geq 0$. Functionally, $\beta_i$ is the variable describing how much agent $i$ ``likes'' the good. $\alpha_{ij}$ and $\beta_i$ values were chosen according to the following three parameters:
\begin{enumerate}
    \item We parameterized instances with an overall ``prosocial utility budget'' $B$, such that for each agent $i$, $\sum_{j \neq i} \alpha_{ij} = B$. We always set $\alpha_{ii} = 1$ for all $i$. Varying this parameter helps us understand how welfare changes when agents behave more or less pro-socially.
    \item A parameter $z$ represents the percent of each agent's prosocial budget allotted to the other agents in their group. $B$ and $z$ completely determine prosocial utilities for every agent. Varying this parameter helps us understand how welfare changes when group memberships are more or less correlated with underlying prosocial utilities.
    \item Lastly, we let each group draw their $\beta_i$ values from a truncated normal distribution (over [0,1]) with a fixed mean specific to that group and a variance $\sigma^2$ that was changed over the course of the experiments.
    We fixed the means of these distributions evenly over the range [0,1], so that group 1's distribution had a mean of 0, group 2's had a mean of 0.25, etc.\footnote{The full list of means is $[0, 0.25, 0.5, 0.75, 1]$.} Varying $\sigma^2$ helps us understand how welfare changes when group membership is correlated with one's utility for a public good.
\end{enumerate}

We varied $B$, $z$, and $\sigma^2$ over the course of our experiments. For each experiment, we performed 50 trials. For each trial, we set $\alpha_ij$ and $\beta_i$ values according to the appropriate parameters and calculated the Nash equilibrium donation amounts under QF, CO-QF, and the direct donation mechanism. We calculated the QF and CO-QF donation amounts via code that numerically solved the system of equations corresponding to that mechanism's Nash equilibrium. The Nash equilibrium under direct donations has a simple closed-form solution where the agent with the highest value for $\alpha_i \cdot \beta$ is the only one who donates. 

\subsection{Simulation Results}

Our figures show each mechanism's USW approximation ratio over increasing pro-sociality budgets $B$. Two separate figures show results under different values of $\sigma^2$. In each figure, we plot welfare guarantees under two different extreme values for $z$. At one extreme we set $z = 4/24$: each agent has equal prosocial utility for all other agents regardless of group. At the other extreme we set $z=1$: agents only have prosocial utility for members of their group.\footnote{Practically, for any choice of mechanism, $B$, and $\sigma^2$, approximation ratios changed monotonically with $z$, so the extreme choices we plot do reflect upper and lower bounds on welfare.} 

Our results are shown in Figure~\ref{fig:simulation}.\footnote{Our plots do not show error bars since they are too small to visually discern. Across all of our parameter choices, the standard error was at most $0.008$ percentage points.}
\begin{figure}
    \centering
    \includegraphics[width=\linewidth]{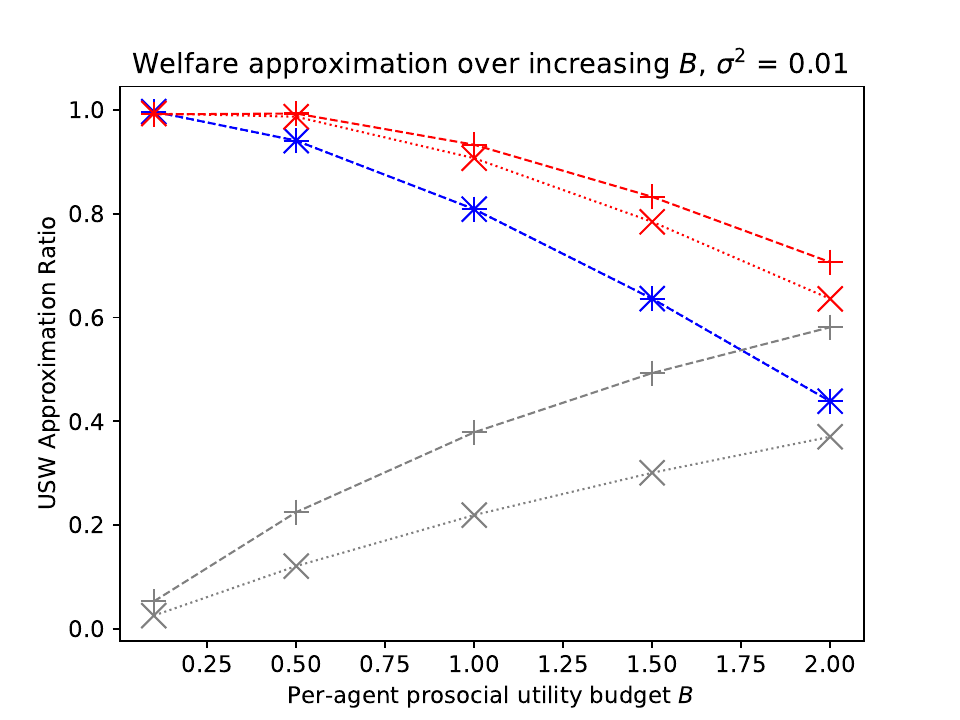}
    \includegraphics[width=\linewidth]{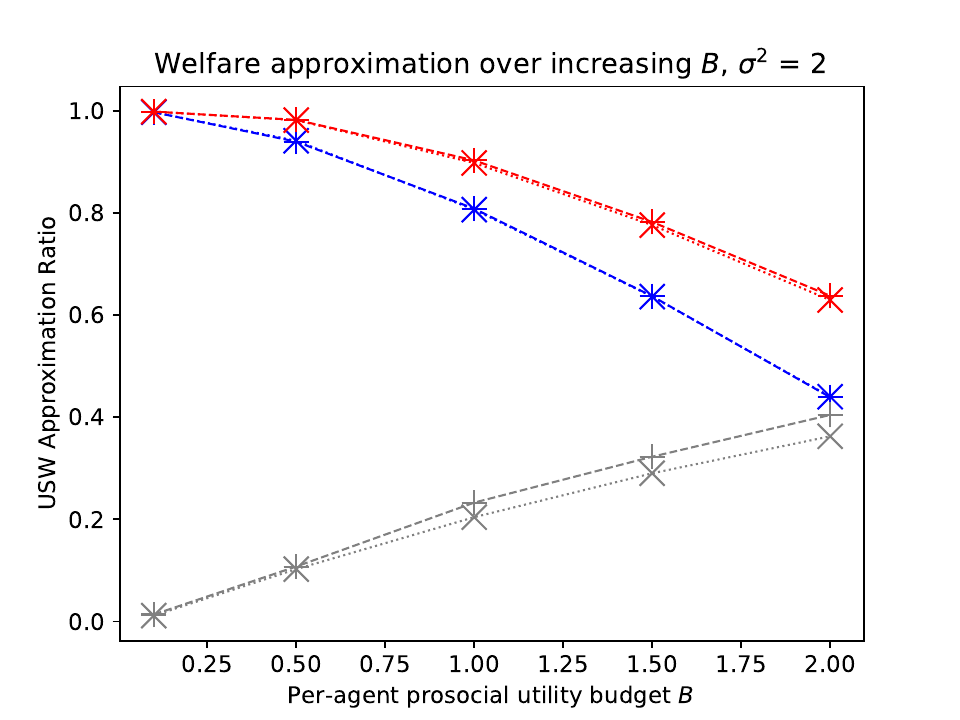}
    \includegraphics[width=0.5\linewidth]{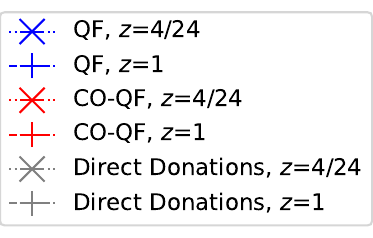}
    \caption{Welfare approximations of QF, CO-QF, and direct donations for varying prosocial budgets. Each data point shows the average approximation ratio of the maximum USW over 50 trials.}
    \label{fig:simulation}
\end{figure}
These results suggest that CO-QF is capable of providing better welfare guarantees than QF and direct donations over the entire range of pro-social utility budgets $B$\footnote{We tested with $B$ values of 0.1, 0.5, 1, 1.5, and 2. Note that at our upper limit of $B=2$ corresponds to an agent caring \textit{twice} as much about all other agents combined than they do about themselves} that we tested. As $B$ approaches 0 QF's approximation ratio approaches 1 as expected.

Figure~\ref{fig:simulation} also suggests that CO-QF is robust to changes in $z$, our parameter denoting how much of one's pro-social utility budget is allotted to agents in their group. The groups CO-QF takes as input are at best estimations, and these simulations suggest that they can indeed be estimations without severely compromising result quality. Lastly, Figure~\ref{fig:simulation} suggests that the performance gap between CO-QF and QF does not depend too much on the correlation between group membership and valuation for the public good. 


Due to our choices of fixed parameters (e.g. the number of agents and general shape of utility functions), these simulations cannot establish the complete landscape of welfare guarantees afforded (or not afforded) by CO-QF. However, they do suggest that 1) CO-QF provides at least as much welfare as QF, even in the original scenario QF was designed for, 2) CO-QF works well even when the groups it takes as input do not represent underlying pro-social utilities completely accurately, and 3) CO-QF is relatively robust to strong corelations between group memberships and valuations.

\section{Discussion}\label{sec:discussion}

Our data suggests that CO-QF had a positive impact on the Gitcoin community, increasing fairness and economic efficiency and gaining sustained adoption over a period of one and a half years. Even though CO-QF did suffer important drawbacks, those issues are specific to Gitcoin and other settings where the mechanism may be seen as a locus for community engagement. On the other hand, CO-QF's strengths over QF (that is, of being fairer and welfare-increasing) seem likely to matter in a wider range of contexts traditionally of interest to those in the mechanism design space. 

This research simultaneously examined the effect of sociality on welfare and fairness outcomes in public decision making, and prioritized the input of a particular social community. Thus, we hope that the methodology of this paper is consistent with its substance, and taken as a whole we hope it mounts a case that the inherent sociality of humans ought to be prioritized in both algorithm design and evaluation. To that end we see several avenues for future work.

\begin{itemize}
    \item As we detailed above, the set of social groups $G$ may be drawn from a source of ground truth (i.e., a firm may use its departments as $G$). When users are instead free to identify themselves with any groups they see fit, questions of how to incentivize truthful reporting naturally arise. Additionally, note that our experiments fixed a set of social groups and from there calibrated pro-social utilities: one could also endogenize pro-social utilities to understand group formation.
    \item While CO-QF can provide better welfare than QF, it is still sub-optimal in many scenarios. Simultaneously, because of the particular ways that any given community might instantiate CO-QF, it may be more useful to prioritize welfare guarantees that are germane to those specific circumstances. 
\end{itemize}

Recall that Quadratic Funding shares similarities with a voting mechanism called Quadratic Voting. The design principles behind CO-QF can easily be applied to Quadratic Voting as well, as we detail in Appendix~\ref{apdx:co-qv}. In this way, we see CO-QF not just as a mechanism for public goods funding but as a blueprint for all manner of collective decision making, guided by the basic principles of plurality and an understanding of humans as inherently social.



\section*{Conflict of Interest Statement}

The design of the second iteration of CO-QF (i.e. Algorithm~\ref{alg:COQF}) and parts of this research project were carried out while one of the authors was employed at Gitcoin part-time. The other author(s) have no current or past affiliation with Gitcoin.

\bibliographystyle{ACM-Reference-Format} 
\bibliography{references}


\begin{thebibliography}{38}


\ifx \showCODEN    \undefined \def \showCODEN     #1{\unskip}     \fi
\ifx \showISBNx    \undefined \def \showISBNx     #1{\unskip}     \fi
\ifx \showISBNxiii \undefined \def \showISBNxiii  #1{\unskip}     \fi
\ifx \showISSN     \undefined \def \showISSN      #1{\unskip}     \fi
\ifx \showLCCN     \undefined \def \showLCCN      #1{\unskip}     \fi
\ifx \shownote     \undefined \def \shownote      #1{#1}          \fi
\ifx \showarticletitle \undefined \def \showarticletitle #1{#1}   \fi
\ifx \showURL      \undefined \def \showURL       {\relax}        \fi
\providecommand\bibfield[2]{#2}
\providecommand\bibinfo[2]{#2}
\providecommand\natexlab[1]{#1}
\providecommand\showeprint[2][]{arXiv:#2}

\bibitem[Batina and Ihori(2005)]%
        {batina2005public}
\bibfield{author}{\bibinfo{person}{Raymond~G Batina} {and} \bibinfo{person}{Toshihiro Ihori}.} \bibinfo{year}{2005}\natexlab{}.
\newblock \bibinfo{booktitle}{\emph{Public goods: theories and evidence}}.
\newblock \bibinfo{publisher}{Springer Science \& Business Media}.
\newblock


\bibitem[Becker(1974)]%
        {becker1974theory}
\bibfield{author}{\bibinfo{person}{Gary~S Becker}.} \bibinfo{year}{1974}\natexlab{}.
\newblock \showarticletitle{A theory of social interactions}.
\newblock \bibinfo{journal}{\emph{Journal of political economy}} \bibinfo{volume}{82}, \bibinfo{number}{6} (\bibinfo{year}{1974}), \bibinfo{pages}{1063--1093}.
\newblock


\bibitem[Bonnet and Teuteberg(2024)]%
        {bonnet2024decentralized}
\bibfield{author}{\bibinfo{person}{Severin Bonnet} {and} \bibinfo{person}{Frank Teuteberg}.} \bibinfo{year}{2024}\natexlab{}.
\newblock \showarticletitle{Decentralized autonomous organizations: A systematic literature review and research agenda}.
\newblock \bibinfo{journal}{\emph{International Journal of Innovation and Technology Management}} \bibinfo{volume}{21}, \bibinfo{number}{04} (\bibinfo{year}{2024}), \bibinfo{pages}{2450026}.
\newblock


\bibitem[Braun and Clarke(2006)]%
        {braun2006using}
\bibfield{author}{\bibinfo{person}{Virginia Braun} {and} \bibinfo{person}{Victoria Clarke}.} \bibinfo{year}{2006}\natexlab{}.
\newblock \showarticletitle{Using thematic analysis in psychology}.
\newblock \bibinfo{journal}{\emph{Qualitative research in psychology}} \bibinfo{volume}{3}, \bibinfo{number}{2} (\bibinfo{year}{2006}), \bibinfo{pages}{77--101}.
\newblock


\bibitem[Bruce and Waldman(1990)]%
        {bruce1990rotten}
\bibfield{author}{\bibinfo{person}{Neil Bruce} {and} \bibinfo{person}{Michael Waldman}.} \bibinfo{year}{1990}\natexlab{}.
\newblock \showarticletitle{The rotten-kid theorem meets the Samaritan's dilemma}.
\newblock \bibinfo{journal}{\emph{The Quarterly Journal of Economics}} \bibinfo{volume}{105}, \bibinfo{number}{1} (\bibinfo{year}{1990}), \bibinfo{pages}{155--165}.
\newblock


\bibitem[Buterin et~al\mbox{.}(2019)]%
        {buterin2019flexible}
\bibfield{author}{\bibinfo{person}{Vitalik Buterin}, \bibinfo{person}{Zo{\"e} Hitzig}, {and} \bibinfo{person}{E~Glen Weyl}.} \bibinfo{year}{2019}\natexlab{}.
\newblock \showarticletitle{A flexible design for funding public goods}.
\newblock \bibinfo{journal}{\emph{Management Science}} \bibinfo{volume}{65}, \bibinfo{number}{11} (\bibinfo{year}{2019}), \bibinfo{pages}{5171--5187}.
\newblock


\bibitem[Charness and Rabin(2002)]%
        {charness2002understanding}
\bibfield{author}{\bibinfo{person}{Gary Charness} {and} \bibinfo{person}{Matthew Rabin}.} \bibinfo{year}{2002}\natexlab{}.
\newblock \showarticletitle{Understanding social preferences with simple tests}.
\newblock \bibinfo{journal}{\emph{The quarterly journal of economics}} \bibinfo{volume}{117}, \bibinfo{number}{3} (\bibinfo{year}{2002}), \bibinfo{pages}{817--869}.
\newblock


\bibitem[Cheng et~al\mbox{.}(2021)]%
        {cheng2021can}
\bibfield{author}{\bibinfo{person}{Ti-Chung Cheng}, \bibinfo{person}{Tiffany~Wenting Li}, \bibinfo{person}{Yi-Hung Chou}, \bibinfo{person}{Karrie Karahalios}, {and} \bibinfo{person}{Hari Sundaram}.} \bibinfo{year}{2021}\natexlab{}.
\newblock \showarticletitle{" I can show what I really like." Eliciting Preferences via Quadratic Voting}.
\newblock \bibinfo{journal}{\emph{Proceedings of the ACM on Human-Computer Interaction}} \bibinfo{volume}{5}, \bibinfo{number}{CSCW1} (\bibinfo{year}{2021}), \bibinfo{pages}{1--43}.
\newblock


\bibitem[Cheng et~al\mbox{.}(2025)]%
        {cheng2025organize}
\bibfield{author}{\bibinfo{person}{Ti-Chung Cheng}, \bibinfo{person}{Yutong Zhang}, \bibinfo{person}{Yi-Hung Chou}, \bibinfo{person}{Vinay Koshy}, \bibinfo{person}{Tiffany~Wenting Li}, \bibinfo{person}{Karrie Karahalios}, {and} \bibinfo{person}{Hari Sundaram}.} \bibinfo{year}{2025}\natexlab{}.
\newblock \showarticletitle{Organize, Then Vote: Exploring Cognitive Load in Quadratic Survey Interfaces}. In \bibinfo{booktitle}{\emph{Proceedings of the 2025 CHI Conference on Human Factors in Computing Systems}} \emph{(\bibinfo{series}{CHI '25})}. \bibinfo{publisher}{Association for Computing Machinery}, \bibinfo{address}{New York, NY, USA}, Article \bibinfo{articleno}{475}, \bibinfo{numpages}{35}~pages.
\newblock
\showISBNx{9798400713941}
\href{https://doi.org/10.1145/3706598.3714193}{doi:\nolinkurl{10.1145/3706598.3714193}}


\bibitem[Corazzini et~al\mbox{.}(2015)]%
        {corazzini2015donor}
\bibfield{author}{\bibinfo{person}{Luca Corazzini}, \bibinfo{person}{Christopher Cotton}, {and} \bibinfo{person}{Paola Valbonesi}.} \bibinfo{year}{2015}\natexlab{}.
\newblock \showarticletitle{Donor coordination in project funding: Evidence from a threshold public goods experiment}.
\newblock \bibinfo{journal}{\emph{Journal of Public Economics}}  \bibinfo{volume}{128} (\bibinfo{year}{2015}), \bibinfo{pages}{16--29}.
\newblock


\bibitem[Dimitri(2022)]%
        {dimitri2022quadratic}
\bibfield{author}{\bibinfo{person}{Nicola Dimitri}.} \bibinfo{year}{2022}\natexlab{}.
\newblock \showarticletitle{Quadratic voting in blockchain governance}.
\newblock \bibinfo{journal}{\emph{Information}} \bibinfo{volume}{13}, \bibinfo{number}{6} (\bibinfo{year}{2022}), \bibinfo{pages}{305}.
\newblock


\bibitem[Edgeworth(1881)]%
        {edgeworth1881mathematical}
\bibfield{author}{\bibinfo{person}{Francis~Ysidro Edgeworth}.} \bibinfo{year}{1881}\natexlab{}.
\newblock \bibinfo{booktitle}{\emph{Mathematical psychics: An essay on the application of mathematics to the moral sciences}}.
\newblock Number~10. \bibinfo{publisher}{CK Paul}.
\newblock


\bibitem[Flanigan et~al\mbox{.}(2023)]%
        {flanigan2023distortion}
\bibfield{author}{\bibinfo{person}{Bailey Flanigan}, \bibinfo{person}{Ariel~D Procaccia}, {and} \bibinfo{person}{Sven Wang}.} \bibinfo{year}{2023}\natexlab{}.
\newblock \showarticletitle{Distortion under public-spirited voting}.
\newblock \bibinfo{journal}{\emph{Proceedings of the 24th ACM Conference on Economics and Computation}} (\bibinfo{year}{2023}).
\newblock


\bibitem[Freitas and Maldonado(2025)]%
        {freitas2025quadratic}
\bibfield{author}{\bibinfo{person}{Luis~Mota Freitas} {and} \bibinfo{person}{Wilfredo~L Maldonado}.} \bibinfo{year}{2025}\natexlab{}.
\newblock \showarticletitle{Quadratic funding with incomplete information}.
\newblock \bibinfo{journal}{\emph{Social Choice and Welfare}} \bibinfo{volume}{64}, \bibinfo{number}{1} (\bibinfo{year}{2025}), \bibinfo{pages}{43--67}.
\newblock


\bibitem[Green and Viljoen(2020)]%
        {green2020algorithmic}
\bibfield{author}{\bibinfo{person}{Ben Green} {and} \bibinfo{person}{Salom{\'e} Viljoen}.} \bibinfo{year}{2020}\natexlab{}.
\newblock \showarticletitle{Algorithmic realism: expanding the boundaries of algorithmic thought}. In \bibinfo{booktitle}{\emph{Proceedings of the 2020 conference on fairness, accountability, and transparency}}. \bibinfo{pages}{19--31}.
\newblock


\bibitem[Groves and Ledyard(1977)]%
        {groves1977optimal}
\bibfield{author}{\bibinfo{person}{Theodore Groves} {and} \bibinfo{person}{John Ledyard}.} \bibinfo{year}{1977}\natexlab{}.
\newblock \showarticletitle{Optimal allocation of public goods: A solution to the" free rider" problem}.
\newblock \bibinfo{journal}{\emph{Econometrica: Journal of the Econometric Society}} (\bibinfo{year}{1977}), \bibinfo{pages}{783--809}.
\newblock


\bibitem[Hammond(2018)]%
        {hammond2018altruism}
\bibfield{author}{\bibinfo{person}{Peter~J. Hammond}.} \bibinfo{year}{2018}\natexlab{}.
\newblock \bibinfo{booktitle}{\emph{The New Palgrave Dictionary of Economics}}.
\newblock \bibinfo{publisher}{Palgrave Macmillan UK}, \bibinfo{address}{London}, Chapter Altruism, \bibinfo{pages}{261--265}.
\newblock
\showISBNx{978-1-349-95189-5}
\href{https://doi.org/10.1057/978-1-349-95189-5_470}{doi:\nolinkurl{10.1057/978-1-349-95189-5_470}}


\bibitem[Hardin(1968)]%
        {hardin1968tragedy}
\bibfield{author}{\bibinfo{person}{Garrett Hardin}.} \bibinfo{year}{1968}\natexlab{}.
\newblock \showarticletitle{The tragedy of the commons: the population problem has no technical solution; it requires a fundamental extension in morality.}
\newblock \bibinfo{journal}{\emph{science}} \bibinfo{volume}{162}, \bibinfo{number}{3859} (\bibinfo{year}{1968}), \bibinfo{pages}{1243--1248}.
\newblock


\bibitem[Harsanyi(1977)]%
        {harsanyi1977morality}
\bibfield{author}{\bibinfo{person}{John~C Harsanyi}.} \bibinfo{year}{1977}\natexlab{}.
\newblock \showarticletitle{Morality and the theory of rational behavior}.
\newblock \bibinfo{journal}{\emph{Social research}} (\bibinfo{year}{1977}), \bibinfo{pages}{623--656}.
\newblock


\bibitem[Hong and Ryu(2019)]%
        {hong2019crowdfunding}
\bibfield{author}{\bibinfo{person}{Sounman Hong} {and} \bibinfo{person}{Jungmin Ryu}.} \bibinfo{year}{2019}\natexlab{}.
\newblock \showarticletitle{Crowdfunding public projects: Collaborative governance for achieving citizen co-funding of public goods}.
\newblock \bibinfo{journal}{\emph{Government Information Quarterly}} \bibinfo{volume}{36}, \bibinfo{number}{1} (\bibinfo{year}{2019}), \bibinfo{pages}{145--153}.
\newblock


\bibitem[Hylland and Zeckhauser(1979)]%
        {hylland1979mechanism}
\bibfield{author}{\bibinfo{person}{Aanund Hylland} {and} \bibinfo{person}{Richard Zeckhauser}.} \bibinfo{year}{1979}\natexlab{}.
\newblock \bibinfo{booktitle}{\emph{A mechanism for selecting public goods when preferences must be elicited}}.
\newblock \bibinfo{publisher}{John Fitzgerald Kennedy School of Government, Harvard University}.
\newblock


\bibitem[Katell et~al\mbox{.}(2020)]%
        {katell2020toward}
\bibfield{author}{\bibinfo{person}{Michael Katell}, \bibinfo{person}{Meg Young}, \bibinfo{person}{Dharma Dailey}, \bibinfo{person}{Bernease Herman}, \bibinfo{person}{Vivian Guetler}, \bibinfo{person}{Aaron Tam}, \bibinfo{person}{Corinne Bintz}, \bibinfo{person}{Daniella Raz}, {and} \bibinfo{person}{PM Krafft}.} \bibinfo{year}{2020}\natexlab{}.
\newblock \showarticletitle{Toward situated interventions for algorithmic equity: lessons from the field}. In \bibinfo{booktitle}{\emph{Proceedings of the 2020 conference on fairness, accountability, and transparency}}. \bibinfo{pages}{45--55}.
\newblock


\bibitem[Kelter et~al\mbox{.}(2021)]%
        {kelter2021agent}
\bibfield{author}{\bibinfo{person}{Jacob Kelter}, \bibinfo{person}{Andreas Bugler}, {and} \bibinfo{person}{Uri Wilensky}.} \bibinfo{year}{2021}\natexlab{}.
\newblock \showarticletitle{Agent-Based Models of Quadratic Voting}. In \bibinfo{booktitle}{\emph{Proceedings of the 2020 Conference of The Computational Social Science Society of the Americas}}. Springer, \bibinfo{pages}{131--142}.
\newblock


\bibitem[Lalley and Weyl(2018)]%
        {lalley2018quadratic}
\bibfield{author}{\bibinfo{person}{Steven~P Lalley} {and} \bibinfo{person}{E~Glen Weyl}.} \bibinfo{year}{2018}\natexlab{}.
\newblock \showarticletitle{Quadratic voting: How mechanism design can radicalize democracy}. In \bibinfo{booktitle}{\emph{AEA Papers and Proceedings}}, Vol.~\bibinfo{volume}{108}. American Economic Association 2014 Broadway, Suite 305, Nashville, TN 37203, \bibinfo{pages}{33--37}.
\newblock


\bibitem[Li et~al\mbox{.}(2024)]%
        {li2024decentralized}
\bibfield{author}{\bibinfo{person}{Jichen Li}, \bibinfo{person}{Yukun Cheng}, \bibinfo{person}{Wenhan Huang}, \bibinfo{person}{Mengqian Zhang}, \bibinfo{person}{Jiarui Fan}, \bibinfo{person}{Xiaotie Deng}, \bibinfo{person}{Jan Xie}, {and} \bibinfo{person}{Jie Zhang}.} \bibinfo{year}{2024}\natexlab{}.
\newblock \showarticletitle{Decentralized Funding of Public Goods in Blockchain System: Leveraging Expert Advice}.
\newblock \bibinfo{journal}{\emph{IEEE Transactions on Cloud Computing}} (\bibinfo{year}{2024}).
\newblock


\bibitem[Miller et~al\mbox{.}(2024)]%
        {miller2024case}
\bibfield{author}{\bibinfo{person}{Joel Miller}, \bibinfo{person}{Chris Kanich}, {and} \bibinfo{person}{E~Glen Weyl}.} \bibinfo{year}{2024}\natexlab{}.
\newblock \showarticletitle{A Case Study in Plural Governance Design}. In \bibinfo{booktitle}{\emph{Pluralistic Alignment Workshop at NeurIPS}}.
\newblock


\bibitem[Nabben(2023)]%
        {nabben2023governance}
\bibfield{author}{\bibinfo{person}{Kelsie Nabben}.} \bibinfo{year}{2023}\natexlab{}.
\newblock \showarticletitle{Governance by algorithms, governance of algorithms: human-machine politics in decentralised autonomous organisations (DAOs)}.
\newblock \bibinfo{journal}{\emph{PuntOorg International Journal}} \bibinfo{volume}{8}, \bibinfo{number}{1} (\bibinfo{year}{2023}), \bibinfo{pages}{36--54}.
\newblock


\bibitem[Pasquini(2020)]%
        {pasquini2020quadratic}
\bibfield{author}{\bibinfo{person}{Ricardo~A Pasquini}.} \bibinfo{year}{2020}\natexlab{}.
\newblock \showarticletitle{Quadratic Funding and Matching Funds Requirements}.
\newblock \bibinfo{journal}{\emph{arXiv preprint arXiv:2010.01193}} (\bibinfo{year}{2020}).
\newblock


\bibitem[Pasquini(2022)]%
        {pasquini2022optimal}
\bibfield{author}{\bibinfo{person}{Ricardo~A Pasquini}.} \bibinfo{year}{2022}\natexlab{}.
\newblock \showarticletitle{Optimal Allocation of Limited Funds in Quadratic Funding}.
\newblock \bibinfo{journal}{\emph{arXiv preprint arXiv:2207.14775}} (\bibinfo{year}{2022}).
\newblock


\bibitem[Quarfoot et~al\mbox{.}(2017)]%
        {quarfoot2017quadratic}
\bibfield{author}{\bibinfo{person}{David Quarfoot}, \bibinfo{person}{Douglas von Kohorn}, \bibinfo{person}{Kevin Slavin}, \bibinfo{person}{Rory Sutherland}, \bibinfo{person}{David Goldstein}, {and} \bibinfo{person}{Ellen Konar}.} \bibinfo{year}{2017}\natexlab{}.
\newblock \showarticletitle{Quadratic voting in the wild: real people, real votes}.
\newblock \bibinfo{journal}{\emph{Public Choice}}  \bibinfo{volume}{172} (\bibinfo{year}{2017}), \bibinfo{pages}{283--303}.
\newblock


\bibitem[Russo et~al\mbox{.}(2024)]%
        {russo2024bridging}
\bibfield{author}{\bibinfo{person}{Mayra Russo}, \bibinfo{person}{Mackenzie Jorgensen}, \bibinfo{person}{Kristen~M Scott}, \bibinfo{person}{Wendy Xu}, \bibinfo{person}{Di~H Nguyen}, \bibinfo{person}{Jessie Finocchiaro}, {and} \bibinfo{person}{Matthew Olckers}.} \bibinfo{year}{2024}\natexlab{}.
\newblock \showarticletitle{Bridging Research and Practice Through Conversation: Reflecting on Our Experience}. In \bibinfo{booktitle}{\emph{Proceedings of the 4th ACM Conference on Equity and Access in Algorithms, Mechanisms, and Optimization}}. \bibinfo{pages}{1--11}.
\newblock


\bibitem[Samuelson(1954)]%
        {samuelson1954pure}
\bibfield{author}{\bibinfo{person}{Paul~A Samuelson}.} \bibinfo{year}{1954}\natexlab{}.
\newblock \showarticletitle{The pure theory of public expenditure}.
\newblock \bibinfo{journal}{\emph{The review of economics and statistics}} (\bibinfo{year}{1954}), \bibinfo{pages}{387--389}.
\newblock


\bibitem[Seaver(2017)]%
        {seaver2017algorithms}
\bibfield{author}{\bibinfo{person}{Nick Seaver}.} \bibinfo{year}{2017}\natexlab{}.
\newblock \showarticletitle{Algorithms as culture: Some tactics for the ethnography of algorithmic systems}.
\newblock \bibinfo{journal}{\emph{Big data \& society}} \bibinfo{volume}{4}, \bibinfo{number}{2} (\bibinfo{year}{2017}), \bibinfo{pages}{2053951717738104}.
\newblock


\bibitem[Simmel(1890)]%
        {Simmel}
\bibfield{author}{\bibinfo{person}{Georg Simmel}.} \bibinfo{year}{1890}\natexlab{}.
\newblock \bibinfo{booktitle}{\emph{Über sociale Differenzierung. Sociologische und psychologische Untersuchungen}}.
\newblock \bibinfo{publisher}{Duncker \& Humblot}, \bibinfo{address}{Leipzig}.
\newblock


\bibitem[Tilman et~al\mbox{.}(2019)]%
        {tilman2019localized}
\bibfield{author}{\bibinfo{person}{Andrew~R Tilman}, \bibinfo{person}{Avinash~K Dixit}, {and} \bibinfo{person}{Simon~A Levin}.} \bibinfo{year}{2019}\natexlab{}.
\newblock \showarticletitle{Localized prosocial preferences, public goods, and common-pool resources}.
\newblock \bibinfo{journal}{\emph{Proceedings of the National Academy of Sciences}} \bibinfo{volume}{116}, \bibinfo{number}{12} (\bibinfo{year}{2019}), \bibinfo{pages}{5305--5310}.
\newblock


\bibitem[User(2022)]%
        {gitcoin2022how}
\bibfield{author}{\bibinfo{person}{Gitcoin~Admin User}.} \bibinfo{year}{2022}\natexlab{}.
\newblock \bibinfo{booktitle}{\emph{How to Attack and Defend Quadratic Funding}}.
\newblock
\urldef\tempurl%
\url{https://www.gitcoin.co/blog/how-to-attack-and-defend-quadratic-funding}
\showURL{%
\tempurl}


\bibitem[Weyl et~al\mbox{.}(2024)]%
        {weyl2024plurality}
\bibfield{author}{\bibinfo{person}{Glen Weyl}, \bibinfo{person}{Audrey Tang}, {and} \bibinfo{person}{Plurality Community}.} \bibinfo{year}{2024}\natexlab{}.
\newblock \bibinfo{booktitle}{\emph{Plurality: The Future of Collaborative Technology and Democracy}}.
\newblock \bibinfo{publisher}{RadicalxChange}.
\newblock


\bibitem[Wright~Jr(2019)]%
        {wright2019quadratic}
\bibfield{author}{\bibinfo{person}{Del Wright~Jr}.} \bibinfo{year}{2019}\natexlab{}.
\newblock \showarticletitle{Quadratic voting and blockchain governance}.
\newblock \bibinfo{journal}{\emph{UMKC L. Rev.}}  \bibinfo{volume}{88} (\bibinfo{year}{2019}), \bibinfo{pages}{475}.
\newblock


\end{thebibliography}

\appendix
\section{An alternative design for CO-QF}\label{apdx:first design}

In our first pass designing CO-QF, we aimed to provide mathematically strong guarantees that large groups couldn't extract too much money from the subsidy pool. However the resultant algorithm was unusable in practice due to the way it reacted to networks with small world properties. In this section we detail the stronger mathematical guarantees we sought, describe the first pass of CO-QF and show that it satisfies those guarantees, and finally explain why this first algorithm was unsuccessful in practice. 

\subsection{Mathematical Desiderata}

One interesting property of QF is that it provides decreasing returns in the subsidy amount as any one individual increases their contribution: fixing all contributions besides $c_i$, $\FQF - \sum_j c_j$ grows at a rate of $O\bigl(\sqrt{c_i}\bigr)$. However, groups of agents break this property. If a group of agents collectively increase their contribution to some project by $x$, the subsidy amount instead grows at a rate of $O(x)$. 

To see why this is, suppose we have a set $N$ of $n$ agents, and a subset $g \subseteq N$. Then we can re-write the subsidy amount as 
\[
  \biggl(\sum_{i \in N} \sqrt{c_i}\biggr)^2 - \sum_{i \in N} c_i = \sum_{i \neq j, i \in N, j \in N} \sqrt{c_i \cdot c_j}
  \geq \sum_{i \neq j, i \in g, j \in g} \sqrt{c_i \cdot c_j}
\]
If each member of $g$ increases their donation by $x$, then the subsidy amount is at least
\[
   \sum_{i \neq j, i \in g, j \in g} \sqrt{(c_i + x) \cdot (c_j + x)} 
   =  (|g|^2 - |g|) \cdot O(x) 
   = O(x)
\]

In other words, under QF individuals are awarded with sub-linear subsidy growth for an increase to their contribution, but groups are awarded with linear subsidy growth.

In light of this observation, our original desiderata for CO-QF was to retain the sub-linear subsidy increase even when groups donated. This felt appropriate in the sense that we were extending an appealing property that QF had for individuals into the world of groups. The question of how to group agents together depends on context, but we aimed to at least create a mechanism that could satisfy this property relative to any set of groups we fed it as input. Our original desiderata were formally stated as follows. 
\begin{definition}[Original Desiderata for CO-QF]\label{def:desiderata}
    Suppose $G \subseteq 2^N$ is a set of groups of agents. Then 
    CO-QF, parameterized with $G$, should award a funding amount $\FCOQF$ such that
    \begin{enumerate}
        \item For any $i \in N$ contributing $x$, $\FCOQF - \sum_j c_j = O\bigl(\sqrt{x}\bigr)$
        \item For any $g \in G$, if the members of $g$ all contribute $x$, then $\FCOQF - \sum_j c_j = O\bigl(\sqrt{x}\bigr)$
    \end{enumerate}
\end{definition}

\subsection{First Definition of CO-QF and Proof}

Given a set of agents $N$ and a set of groups $G$, let $T_i$ denote the set of groups of which agent $i$ is a member. For an agent $i$ and a group $h$, let $k_{i,h}$ denote a measure of social connection between agent $i$ and group $h$. In particular, Define $k_{i,h}$ as 
\begin{align*}
    k_{i,h} = 
    \begin{cases}
        1 &\hbox{if}\;\; T_i \cap \bigcup_{j \in h} T_j = \emptyset\\
        2 &\hbox{otherwise}
    \end{cases}
\end{align*}
Then funding amount under CO-QF was originally defined as 
\[
    \FCOQFp = \sum_{i \in N} c_i + \sum_{g\neq h, g \in G, h \in G} \sqrt{\sum_{i \in g} \frac{c_i^{1/k_{i,h}}}{|T_i|}\cdot \sum_{j \in h}\frac{c_j^{1/k_{j,g}}}{|T_j|}}
\]

Note that if $G = \{ \{i\} \;|\;  i \in N\}$ then this iteration of CO-QF also gives the same funding as QF. This iteration of CO-QF satisfies the desiderata of Definition~\ref{def:desiderata}, as we show below. 

\begin{theorem}
    The original design of CO-QF satisfies the desiderata of Definition~\ref{def:desiderata}.
\end{theorem}

\begin{proof}
    First observe that 
    \begin{align*}
        \FCOQFp - \sum_{i}c_i = \sum_{g\neq h, g \in G, h \in G} \sqrt{\sum_{i \in g} \frac{c_i^{1/k_{i,h}}}{|T_i|}\cdot \sum_{j \in h}\frac{c_j^{1/k_{j,g}}}{|T_j|}}
    \end{align*}
    Since there are a constant number of terms in this sum, it suffices to show that each individual term of the form 
    \begin{align}\label{math:interaction term}
    \sqrt{\sum_{i \in g} \frac{c_i^{1/k_{i,h}}}{|T_i|}\cdot \sum_{j \in h}\frac{c_j^{1/k_{j,g}}}{|T_j|}}
    \end{align}
    satisfies the asymptotic requirements of Definition~\ref{def:desiderata}.
    
    \textit{(Property 1)} Let $g$ and $h$ be members of $G$. If $i \notin g$ and $i \notin h$, then Eq.~\ref{math:interaction term} is constant with respect to $c_i$. If $i \in g$ and $i \in h$, then Eq.~\ref{math:interaction term} is $O\bigl(\sqrt{c_i}\bigr)$ since $k_{i,g} = 2$ and $k_{i,h} = 2$. Disregarding constants, Eq.~\ref{math:interaction term} will be $O\bigl(\sqrt{\sqrt{c_i} \cdot \sqrt{c_i}} \bigr)= O\bigl(\sqrt{c_i}\bigr)$. If $i$ is in one of $g$ or $h$, then Eq.~\ref{math:interaction term} is either $O\bigl(\sqrt{c_i}\bigr)$ or $O\bigl(c_i^{1/4}\bigr)$, both of which satisfy property 1.

   \textit{(Property 2)} Let $g$ be a member of $G$ and suppose all members of $g$ contribute $x$. Now consider a term of the form outlined in Eq.~\ref{math:interaction term} with respect to two groups $h_1$ and $h_2$ (possibly, one or both of these groups are $g$ itself). If $|g \cap h_1| = 0$ and $|g \cap h_2| = 0$ than the term is a constant with respect to $x$. Suppose $g$ overlaps with one group but not the other, w.l.o.g. assume $|g \cap h_1| > 0$ and $|g \cap h_2| = 0$. Then we have 
\begin{align*}
    &\sqrt{\sum_{i \in h_1} \frac{c_i^{1/ k_{i,h_2}}}{|T_i|}\cdot \sum_{j \in h_2}\frac{c_j ^{1/k_{j,h_1}}}{|T_j|}} 
    \\%
    \leq &\sqrt{\biggl(|g| * x + \sum_{i \in h_1 \setminus g} \frac{c_i^{1/k_{i,h_2}}}{|T_i|}\biggr)\cdot \sum_{j \in h_2}\frac{c_j ^{1/ k_{j,h_1}}}{|T_j|}}
    = O(\sqrt{x})
\end{align*}
        
   where the inequality comes from noting that at most all members of $g$ are also in $h_1$, and from noting that $k_{i,h_2} \geq 1$ for all $i \in g$ and $|T_i| \geq 1$ for all $i$. The equality comes from removing constants with respect to $x$.
   
   Finally, suppose that $|g \cap h_1| > 0$ and $|g \cap h_1| > 0$. In this case, we will have $k_{i,h_2} = 2$ for all $i \in g \cap  h_1$ and $k_{j,h_1} = 2$ for all $j \in g \cap h_2$, so the term becomes
   \begin{align*}
        &\sqrt{\sum_{i \in h_1} \frac{c_i^{1/ k_{i,h_2}}}{|T_i|}\cdot \sum_{j \in h_2}\frac{c_j ^{1/k_{j,h_1}}}{|T_j|}} 
        \\%
        \leq 
        &\sqrt{\biggl(|g| * \sqrt{x} + \sum_{i \in h_1 \setminus g} \frac{c_i^{1/k_{i,h_2}}}{|T_i|}\biggr)\cdot \biggl(|g| * \sqrt{x} + \sum_{i \in h_1 \setminus g} \frac{c_i^{1/k_{i,h_2}}}{|T_i|}\biggr)}
        \\%
        = &O\biggr(\sqrt{\sqrt{x} \cdot \sqrt{x}})\biggl)
        = O(\sqrt{x})
    \end{align*}
    Where the inequality comes from noting that at most, all members of $g$ are also in $h_1$ and $h_2$, and $|T_i| \geq 1$ for all $i$. The equalities come from disregarding constants with respect to $x$ and simplifying.
\end{proof}

\subsection{Moving Away From the First Design} 

Like the second design for CO-QF detailed in the main text of the paper, the first design iterates over pairs of groups. However, the first design of CO-QF can also attenuate contributions by taking their square root an additional time during this iteration. For example, suppose the algorithm is iterating over the pair of groups $(g,h)$. Then for an agent $i \in g$, the first design of CO-QF takes an additional square root of $c_i$ if $i$ is in $h$ \textit{or} is in any group with a member of $h$. While taking additional square roots did help us guarantee that subsidies would grow sub-linearly in response to group contributions, in practice the groups we worked with often had enough overlap that every single contribution would be square rooted twice. Attenuating every contribution in this way led to funding results that were very even. 

Unfortunately we did not collect qualitative interview data on opinions about this first mechanism, so we are unable to report on its reception with certainty. However, casual conversations at the time with Gitcoin employees, round managers, and grantees broadly bore the same theme: these stakeholders viewed the exceedingly flat funding results as a cause for concern, since it was not clear how the ``good'' projects were being prioritized, if at all. 

Therefore, in the second pass of desiging CO-QF we attempted to retain some of the spirit of the first version, while not attenuating contributions as much. Unfortunately, this did mean relinquishing the guarantees of Definition~\ref{def:desiderata}, but we see this as another interesting example of the tensions that arise between mathematical desiderata and the practical circumstances of a community.

\section{Other Ways of Choosing Social Groups}\label{apdx:otherG}

Before settling on using donation data to instantiate $G$ and $w$, we tried two different approaches.

First, since many Gitcoin donors use blockchain tools and applications, the platform looked to public blockchain-based sources of data. The platform specifically looked at using data from POAP (Proof Of Attendance Protocol) \footnote{\url{https://poap.xyz/}} and Guild \footnote{\url{https://guildprotocol.io/}}. Both of these tools are allow users to publish public, cryptographically verifiable attestations to various social arrangements on the Ethereum blockchain. The chief difference between the two is that POAP tends to be used for logging information related to specific events (i.e., who attended a certain conference or hackathon), whereas Guild tends to be used to register members of an organization. However, only around 20\% of Gitcoin donors had published any records relating to POAP or Guild on the Ethereum blockchain, so there was not enough data to work with.

Next, the platform experimented with using donation data, but in a different way. Within a round, each user was assigned a binary string corresponding to the set of projects they donated to. So under this scheme, in a round with $m$ projects, $G$ is the set of binary vectors of length $m$, and each user $i$ is assigned to the exactly one group, namely the group
\[
(\textbf{1}\{c_{i,p} > 0\})_{1\leq p \leq m}
\]
where $\textbf{1}\{c_{i,p} > 0\}$ is an indicator variable that is $1$ if $i$ donated to project $p$, and $0$ otherwise. 

This scheme was discarded because of a security flaw. A group of coordinating individuals (or Sybils) all aiming to support a project $p$ can easily appear diverse from each other with the following strategy: each agent donates a significant amount to $p$, and some small amount to a unique set of other projects. Since $G$ has $2^m$ elements, any group of colluding agents of size less than $2^{m-1}$ can easily donate such that they are all put in different groups, thereby executing an attack comparable in severity to what is possible under standard QF. While preventing Sybil attacks and fraud was not our main goal in designing CO-QF, this vulnerability was too severe to be overlooked.

\section{An Extension Quadratic Voting}\label{apdx:co-qv}

Quadratic Voting (QV) is a voting mechanism closely related to QF, and the modifications to QF discussed in this paper can also apply to QV. In QV agents spend ``voice credits'' on an issue, which the mechanism converts to effective votes at quadratic costs. If agent $i$ spends $v_i$ voice credits towards an issue, then the effective number of votes given to that issue is 
\[
    \sum_{i\in N}\sqrt{v_i}
\]
Notice that if we let the $v_i$ values denote contributions (i.e., letting $c_i = v_i$), then the above formula can be re-written as the root of the total amount of funding under vanilla QF:
\[
    \sum_{i\in N}\sqrt{v_i}=
    \sqrt{\left( 
    \sum_{i \in N}\sqrt{c_i}
    \right)^2}
\]
At first glance, this might seem like a needlessly complicated way to re-write the formula for QV. However, the point is that instead of putting the formula for vanilla QF inside the square root on the RHS of the above equality, we could plug in CO-QF instead. In other words, this way of generalizing QV (as the root of the funding amount given by some QF-like formula) opens the door to the use of CO-QF for voting as well.

\end{document}